\documentclass[11pt,envcountsame,letterpaper]{llncs}
\usepackage{amsmath} 
\usepackage{amssymb}
\usepackage[vlined]{algorithm2e} 
\usepackage{algorithmic}
\usepackage{enumitem}
\usepackage{url}
\usepackage{color}
\usepackage[margin=3.6cm]{geometry} 

\title{\large Deterministic Leader Election Takes $\Theta(D+\log n)$ Bit Rounds
\thanks{This research has been partially supported by ANR projects DESCARTES and ESTATE (resp. ANR-16-CE40-0023 and ANR-16-CE25-0009-03). A preliminary subset of this work 
appeared in the proceedings of DISC 2016~\cite{CMRZ16a}.
}
}
\author{A. Casteigts, Y. M\'etivier, J.M. Robson, and A. Zemmari}
\institute{Universit\'e de Bordeaux - Bordeaux INP
 LaBRI, UMR CNRS 5800\\ 351 cours de la
  Lib\'eration, 33405 Talence, France\\ 
\{acasteig, metivier, robson, zemmari\}@labri.fr } 

\begin{document}

\newcommand{\STT}{\ensuremath{{\cal STT}}\xspace}

\maketitle 
\date{today}
\begin{abstract}
  Leader election is, together with consensus, one of the most central
  problems in distributed computing. This paper presents a distributed
  algorithm, called \STT, for electing deterministically a leader in
  an arbitrary network, assuming processors have unique identifiers of
  size $O(\log n)$, where $n$ is the number of processors. It elects a
  leader in $O(D +\log n)$ rounds, where $D$ is the diameter of the
  network, with messages of size $O(1)$. Thus it has a bit round
  complexity of $O(D +\log n)$. This substantially improves upon the
  best known algorithm whose bit round complexity is $O(D\log n)$. In
  fact, using the lower bound by Kutten et al. (2015) and a
  result of Dinitz and Solomon (2007), we show that the bit round
  complexity of \STT is optimal (up to a constant factor), which is a
  significant step forward in understanding the interplay between time
  and message optimality for the election problem. Our algorithm
  requires no knowledge on the graph such as $n$ or $D$, and the
  pipelining technique we introduce to break the $O(D\log n)$ barrier
  is general.
\end{abstract}
\section{Introduction}
\label{sec:introduction}
The election problem in a network consists of distinguishing a unique node, the leader, which
 can subsequently act as coordinator, initiator, 
and more generally performs distinguished operations in the network
 (see \cite{Tanenbaum} p. 262). 
Indeed, once a leader is established, many problems become simple, making election a common building
block
in distributed computing. Election is probably the most studied task (together with consensus) in the distributed computing literature~\cite{DMR08},
starting with the works of Le Lann~\cite{LeLann} and Gallager~\cite{Gallager} in the late 70's. 

A distributed algorithm solves the
election problem if it always terminates and in the final
configuration exactly one process (or node) is in the \emph{elected} state 
and all
others are in the \emph{non-elected} state. It is also required that 
once a process becomes elected or non-elected, it remains so for the rest of the execution.
The vast body of literature on election (see~\cite{Attiya,Lynch,Santoro,Tel} and references therein) actually covers a number of different topics, which can be grouped according to three main directions: i) The feasibility of deterministic election in {\em anonymous} networks, starting with the seminal paper of Angluin~\cite{Angluin} and the key role of coverings (i.e. graph homomorphisms that prevent symmetry breaking, and thereby the uniqueness of a leader); ii) The complexity of deterministic election in {\em identified} networks (i.e. every processor has a unique identifier); and iii) The complexity of {\em probabilistic} election in anonymous or identified networks (identifiers play secondary roles here).


The present work is in the second category, that is, we assume that each node has a unique identifier which is a positive integer of size $O(\log n)$, where $n$ is the number of nodes. The network is multi-hop and nodes communicate using messages in synchronous rounds.
The exact complexity of deterministic leader election in this setting has proven elusive for decades and even simple questions remain open~\cite{KPPRT15}. We review here the most relevant results and challenges around this problem (the reader is referred to the dedicated section for more content).
In the case of logarithmic-size messages (i.e. messages of size $O(\log n)$), we know since Peleg~\cite{Peleg90} that $O(D)$ rounds are sufficient to elect a leader in arbitrary networks, where $D$ is the diameter of the network.
This was recently proven optimal by Kutten et al.~\cite{KPPRT15}
using a very general $\Omega(D)$ lower bound (that applies even in the probabilistic setting). Independently, Fusco and Pelc~\cite{FP15} showed that the 
time complexity
of leader election is
$\Omega(D+\lambda)$ where $\lambda$ is the smallest depth at which some node has a unique view, called the {\em level of symmetry} of the network. (The view at depth $t$ from a node is the tree of all paths
of length $t$ originating at this node.) If nodes have unique identifiers, then $\lambda=0$, which implies the same $\Omega(D)$ bound as in~\cite{KPPRT15}.

Regarding message complexity, 
Gallager~\cite{Gallager} presents
 the first election 
algorithm for general graphs with $O(m+n\log n)$ messages, 
where $m$ is the number of edges, and a running time of $O(n\log n)$.
On the negative side, Burns~\cite{burns1980} proves a $\Omega(n\log n)$ lower bound and Kutten et al.~\cite{KPPRT15} a $\Omega(m)$ lower bound which applies even if $n$ is known and the algorithm is randomized. Put together, both lower bounds yield a matching $\Omega(m + n\log n)$ number of messages. (Santoro~\cite{Santoro84} also proves a $\Omega(m+n\log n)$ lower bound for the more specific problem of finding the maximum ID, in a deterministic setting with $n$ unknown.)

A few years after Gallager~\cite{Gallager}, Awerbuch~\cite{Awer87} presents an algorithm whose message complexity is again $O(m+n\log n)$, but running time is taken down to $O(n)$.

A number of questions remain open for election. Peleg asks in~\cite{Peleg90}
whether an algorithm could be both optimal in time and in number of messages.
The answer depends on the setting, but remains essentially open~\cite{KPPRT15}.
In the conclusion of their paper, Fusco and Pelc \cite{FP15} also observe that
it would be interesting to investigate other complexity
measures for the leader election problem, such as {\em bit complexity}.
This measure can be viewed as a natural
extension of communication complexity (introduced by Yao \cite{Yao})
to the analysis of tasks in a distributed setting.

Following~\cite{KOSS}, the bit round complexity of an algorithm $\cal A$ is the total number of {\it bit rounds} it takes for $\cal A$ to terminate, where a bit round is a round with single bit messages. 
This measure has become popular recently, as it captures into {\em a single quantity} aspects that relate both to time and to the amount of information exchanged.
In this framework, the time-optimal algorithm of Peleg~\cite{Peleg90} results in a bit round complexity of $O(D\log n)$ (i.e. $O(D)$ rounds with $O(\log n)$ message size), and the message-optimal algorithm of~\cite{Awer87} results in a $O(n\log n)$ bit round complexity (i.e. $O(n)$ time with $O(\log n)$ message size). More recent approaches such as~\cite{FSW14} in the beeping models (therefore exchanging single bits per rounds) still remain at a $O(D\log n)$ bit round complexity.

In this paper, we present the first algorithm whose bit round complexity breaks the $O(D\log n)$ barrier, using essentially a new pipelining technique for spreading the identifiers. Our algorithm requires only $O(D + \log n)$ bit rounds and works in arbitrary synchronous networks. We show that this is optimal by combining a lower bound from~\cite{KPPRT15} and a recent communication complexity result by Dinitz and Solomon~\cite{DS07}.
This work is thus a step forward in understanding election, and a first (positive) answer to whether optimality can be
achieved
both in time and in the {\em amount} of information exchanged. (As opposed to measuring time on the one hand, and the number of messages {\em of a given size} on the other hand.) Incidentally, our results also
illustrate
the benefits of studying optimality under the unified lenses of bit round complexity.

\subsection{Contributions}

We present an election algorithm \STT, having time complexity of $O(D+\log n)$ with messages of size $O(1)$, where $D$ is the diameter of the network.
Algorithm \STT solves the {\em explicit} (i.e. strong) variant of the problem defined in~\cite{KPPRT15}, namely, the identifier of the elected node is eventually known to all the nodes. It also fulfills requirements from~\cite{DS07}, such as ensuring that every non-leader node knows which local link is in direction of the leader, and these nodes learn the maximal id network-wide ({\em MaxF}), as a by-product of electing this specific node in the {\em explicit} variant.

The global architecture of our algorithm follows a (now) classical principle, similar to that used e.g. by Gallager~\cite{Gallager} or Peleg~\cite{Peleg90}. It consists of a competition of spanning tree constructions that works by extinction of those trees originating at nodes with lower identifiers (see also Algorithm~4 in~\cite{Attiya} and discussion therein). Eventually, a single spanning tree survives, whose root is the node with highest identifier. This node becomes elected when it detects termination (recursively from the leaves up the root). Here the difficulty (and thus main contribution) arises from designing such algorithms with the extra constraint that only constant size messages are used. Of course, one might simulate $O(\log n)$-size messages in the obvious way paying $O(\log n)$ bit rounds for each message. But then, the bit round complexity remains $O(D \log n)$. In contrast, by introducing new pipelining techniques whose applicability extends the scope of the sole election problem, we take the complexity down to $O(D + \log n)$.

For ease of exposition, the $\cal STT$ algorithm is split into three components, whose execution is joint in a specific way.

\begin{enumerate}
\item A spreading algorithm $\cal S$ which pipelines the maximal identifier bitwise to every node, in a mix of battles (comparisons), conquests (progress of locally higher prefixes), and correction waves of bounded amplitude;
\item A spanning tree algorithm that executes in parallel of $\cal S$ and whose union with $\cal S$ is denoted $\cal ST$. It consists of updating the tree relations based on what neighbour brought the highest prefix so far;
\item A termination detection algorithm that executes in parallel of $\cal ST$ and whose union with $\cal ST$ is denoted $\cal STT$. This component enables the node with highest identifier (and only this one) to detect termination of
the spanning tree construction of which it is itself the root.
\end{enumerate}

An extra component can be added to broadcast a (constant size) termination signal from the root down the tree, once election is complete. This component is trivial and therefore not described here.

\subsubsection{Lower Bound:}

Dinitz and Solomon~\cite{DS07} prove a lower bound (Theorem~1 below) on the leader election problem among two nodes.
\begin{theorem}[\cite{DS07}]
Let $M$ be an integer such that $M\geq 2$.
Let $G$ be the graph with two nodes linked by an edge, and suppose that each node
has a unique identifier taken from the set $Z_M=\{0,\cdots, M\}$. The bit
round complexity
of the Leader task and of the MaxF version is exactly 
$2\lceil \log_2((M+2)/3.5)\rceil$.
\end{theorem}
This theorem implies that  the time complexity of an election algorithm with
messages of size $O(1)$ (bit round complexity) is $\Omega(\log n)$.

On the other hand, the lower bound by Kutten et al. in~\cite{KPPRT15}, establishing that $\Omega(D)$ time is required with logarithmic size messages, obviously extends to constant size messages. Put together, these results imply that the bit  complexity of leader election
with messages of size $O(1)$ and identifiers of size $O(\log n)$ is $\Omega(D+\log n)$, which makes our algorithm bit round optimal (up to a constant factor).

In fact, the lower bound holds for arbitrary sizes $|id|$ of identifiers (necessarily larger than $\log n$, though, since they must be unique). Likewise, the complexity of our algorithm is expressed relative to identifiers of arbitrary sizes (see Theorem~\ref{th:main}). Hence, the bit round complexity of the election problem is in fact $\Theta(D + |id|)$.
Table~\ref{tab:result} summarises these elements, taking $|id|=O(\log n)$ as the most common (illustrative) value.

\begin{table}
\label{tab:result}
\centering
\begin{tabular}{|@{\,}l@{\,}|@{\,}c@{\,}|@{\,}c@{\,}|@{\,}c@{\,}|@{\,}c@{\,}|}
\hline
& Time & Number of messages & Message size & Bit round complexity \\
\hline
Awerbuch  \cite{Awer87}&  $O(n)$ & $\Theta(m+n\log n)$  & $O(\log n)$ & $O(n\log n)$
\\
\hline
Peleg  \cite{Peleg90}&  $\Theta(D)$  & $O(D\, m)$ & $O(\log n)$ & $O(D\log n)$ \\
\hline
This paper & $O(D+\log n)$   & $O((D+\log n) m)$ & $O(1)$ &$\Theta(D+\log n)$ 
 \\\hline
\end{tabular}
\medskip
\caption{Best known solutions in terms of time and number of messages, compared to our algorithm.}
\end{table}

\subsubsection{Outline:} After general definitions in Section~\ref{sec:model}, we present the three components of the algorithm: the spreading algorithm ${\mathcal S}$
(Section~\ref{sec:S}), its joint use with the spanning tree algorithm (${\mathcal ST}$, Section~\ref{sec:ST}), and the adjunction of termination detection (${\mathcal STT}$, Section~\ref{sec:STT}).
Further related works on the leader election problem are presented in Section~\ref{sec:related}.
We conclude in Section~\ref{sec:conclusion} with some remarks.

\section{Model and definitions}
\label{sec:model}

This section presents the network model (synchronous message passing, unique identifiers) and give the main definitions used throughout the paper regarding graph theory, language theory, and bit complexity.

\subsection{The Network}
We consider a failure-free message passing model
in a point-to-point
communication network described by a   connected graph 
$G=(V,E)$
where the nodes $V$ represent network processes (or nodes) and the edges $E$
represent bidirectional communication channels. Processes communicate
by message passing: a process sends a message to another by depositing
the message in the corresponding channel. 

Let $n$ be the size of $V$.
We assume that each node  $u$ is identified by a unique
positive integer of $O(\log n)$ bits, called identifier and denoted $Id_u$
(in fact, $Id_u$ denotes both the identifier and its {\em binary representation}).
 We do not assume any global knowledge on
the network, not even the size or an upper bound on the size, and the nodes do not require position or distance information. 
Every node is equipped with a port numbering function (i.e. a bijection between the set of incident edges $I_u$ and the integers in $[1,|I_u|]$), which allows it to identify which channel a message was received from, or must be sent to.
Two nodes $u$ and $v$ are said to be neighbours if they can communicate
through a port.

Finally, we assume the system is
fully synchronous, namely, all processes start at the same time
and time proceeds in synchronised rounds composed of the following three
steps: 
\begin{enumerate}
\item Send messages to (some of) the neighbours, 
\item Receive
messages from (some of) the neighbours, 
\item Perform local
computation.
\end{enumerate}
  The time complexity of an algorithm is the number of such rounds needed to complete the execution in the worst case.

\subsection{Further definitions}



The paper uses a number of definitions from graph theory and formal language theory. Although most readers may be familiar with them, we recall the most important ones. Then we define the bit round complexity.

\paragraph{Definitions on graphs:}  These definitions are selected from~\cite{Rosenpress} (Chapter 8). A tree is a connected acyclic graph. A rooted tree is a tree with one
 distinguished node, called the root, in which all edges are implicitly 
directed away from the root. 
A spanning tree of a connected graph $G=(V,E)$ is a tree $T=(V,E')$ such that
$E'\subseteq E$.
A forest is an acyclic graph. A spanning forest of a graph $G=(V,E)$ 
is a forest whose node set is $V$ and edge set is a subset of $E$.
A rooted forest is a forest such that each tree of the forest is rooted.
A child of a node $u$ in a rooted tree is an immediate
successor of $u$ on a path from the root.
A descendant of a node $u$ in a rooted tree is $u$ itself or any node that
is a successor of $u$ on a path from the root.
The parent of a node $u$  in a rooted tree is a node that is the immediate
predecessor of $u$ on a path to $u$ from the root.

\paragraph{Definitions on languages:} These definitions are selected from~\cite{Rosenpress} (Chapter 16).
Let $A$ be an alphabet, $A^*$ is
 the set of all words over $A$, the empty word is denoted by $\epsilon$.
If $x$ is a
non-empty
word of length $p$ over the alphabet $A$ then
$x$ can be written as the concatenation of $p$ letters, i.e.,
$x=x[1]x[2]\cdots x[p]$ with each $x[i]$ in $A$.
If $a\in A$ and $i$ is a positive integer then $a^i$ is the 
 concatenation  $i$ times of  the letter $a$.
For two words $x$ and $y$ over alphabet $A$, $x$
is said to be a prefix ({\it resp.} proper prefix) of $y$ if there exists a word ({\it resp.} non-empty word) $z$
such that $y=xz$.

\paragraph{Bit round complexity:}
The bit complexity in general may be viewed as a natural
extension of communication complexity (introduced by Yao \cite{Yao})
to the analysis of tasks in a distributed setting. An introduction to the area can be found in Kushilevitz and Nisan \cite{KN}. 
In this paper, we follow the definition from~\cite{KOSS}, that is, the bit round complexity of an algorithm $\cal A$ is the total number of {\it bit rounds} it takes for $\cal A$ to terminate, where a bit round is a synchronous round with single bit messages. 
This measure captures into a single quantity aspects that relate both to time and to the amount of information exchanged. Other definitions are considered in the literature, in~\cite{BNNN90,BMW,BT90,DMR08} the bit complexity is the total number of bits sent until global termination. In~\cite{sw11}, it is the maximum number of bits sent through a same channel.
In both variants, silences may convey
much information, which is why we consider the definition from~\cite{KOSS} in terms of {\em round} complexity as more comprehensive.

\section{A spreading algorithm}
\label{sec:S}
We present a distributed spreading
algorithm 
using only messages of size $O(1)$  which allows each node to know
the highest identifier among the set of all identifiers with a time
complexity of $O(D+\log n)$, where $D$ is the diameter of $G$. This algorithm is the main component of the \STT algorithm, standing for the 
${\cal S}$ in the acronym.

\subsection{Preamble}
 
Given
the binary representation $Id$ of an identifier,
we define 
$\alpha(Id)$ as the word
$$
\alpha(Id)=1^{|Id|} 0  Id.
$$

For instance, for the integer 23, $Id=10111$ and $\alpha(Id)=11111010111$. This encoding has the nice property that it extends the natural order $<$ of integers into a lexicographic order $\prec$ on their $\alpha$-encoding.



  \label{rem:lexicographic}
If $u$ and $v$ are two nodes with identifiers $Id_u$ and
$Id_v$,
$$
Id_u <  Id_v \Leftrightarrow  \alpha(Id_u) \prec \alpha(Id_v).
$$

As a result, the order between two identifiers $Id_u$ and $Id_v$
is the order induced by the first letter which differs in $\alpha(Id_u)$ and $\alpha(Id_v)$.
This property is key to our algorithm, in which the spreading of identifiers progresses bitwise based on prefix comparisons.

\subsection{The  algorithm $\cal S$}
We describe here the spreading component of the algorithm, {\it i.e.} the $\cal S$ in \STT, whose purpose is to spread the largest identifier network-wide. For simplicity, we present here the algorithm independently from termination detection, which is dealt with in a dedicated section (Section~\ref{sec:STT}).

\subsubsection{Variables:} Each node can be $active$ or $follower$, depending on whether it is still a candidate for becoming the leader (i.e. no higher identifier was detected so far). Each node $u$ also has variables $Y_u$, $Z_u$ and
$Z_{[u]v}$  (one for each neighbour $v$ of $u$) which are words over the alphabet $\{0,1\}$. $Y_u$ is a shorthand for $\alpha(Id_u)$, it is set initially and never changes afterwards. $Z_u$ is a prefix of $Y_w$, for some node $w$ (possibly $u$ itself). It indicates the highest prefix known so far by $u$.
On each node, this variable will eventually converge to the $\alpha$-encoding of the highest identifier. Finally, for each neighbour $v$ of $u$, 
$Z_{[u]v}$ is the latest value of $Z_v$ known to $u$. 

\subsubsection{Initialisation:} Initially every node $u$ is $active$, all the $Z_u$ are set to the empty word $\epsilon$, and the $Z_{[u]v}$ are accordingly set to the empty word.

\subsubsection{Main loop:}  In each round, the algorithm executes the following actions.
\begin{enumerate}[leftmargin=3em]
\item update $Z_u$ based on information received in the previous round,
\item send to all neighbours a signal indicating how $Z_u$ was updated,
\item receive such signals from neighbours,
\item update all the $Z_{[u]v}$ accordingly.
\end{enumerate}

\noindent
The main action is the update of $Z_u$ (step 1).
It depends on the values of $Z_{[u]v}$  for all neighbours $v$ and 
 $Z_u$ itself at the end of the previous round. 
This update is done according to a number of rules. For instance, as long as $u$ remains $active$ and $Z_u$ is a proper prefix of $Y_u$, the update consists in appending the next bit of $Y_u$ to $Z_u$. Most updates are more complex and detailed further below. The three other actions (step $2$, $3$, and $4$ above) only serve the purpose of informing the neighbours as to how $Z_u$ was updated, so that all $Z_{[u]v}$  are correctly updated. In fact, $Z_u$ can only be updated in {\em seven} possible ways, each causing the sending of a (constant size) particular signal among $\{append0, append1, delete1, delete2, delete3, change, null\}$, with following meaning:
\begin{itemize}
\item $append0$ or $append1$: $Z_u$ was updated by appending a single $0$ or a single $1$;
\item $delete1$, $delete2$, or $delete3$: $Z_u$ was updated by deleting one, two or three letters from the end;
\item $change$: $Z_u$ was updated by changing the last letter from $0$ to $1$;
\item $null$: $Z_u$ was not modified.
\end{itemize}

By the end of each round, it holds that $Z_{[u]v}=Z_v$ for all neighbours $v$ of $u$.
Thus from now on, $Z_{[u]v}$  is simply written $Z_v$. Another invariant is that, by the end of each round, if $u$ and $v$ are two neighbours, then $Z_u$ and $Z_v$ must have a common prefix followed, in each case, by at most six letters (see the proof of Lemma 6 below, second item).

We now describe the way $Z_u$ is updated by each node $u$ ({\it i.e.} step 1).
 
\subsubsection{Update of $Z_u$ in each round:}

Let us denote the state of some variable $X$ {\em at the end} of round~$t$ by $X^t$. For instance, we write $Z_u^0=\epsilon$, where round $0$ corresponds to initialisation.
The computation of $Z_u$ at round $t$ results from $u$ being active or follower, and from the values of 
 $Z_u^{t-1}$ and $Z_v^{t-1}$ for all neighbours $v$ of $u$. 
It is done according to the following rules given in order of 
priority, i.e., $R_{1}$ has a higher priority than 
$R_2$. 
(This does not apply between the subrules $R_{1.1}$ and $R_{1.2}$, for which a different criterion is specified.)
Whenever a rule is applied, the subsequent rules are ignored.
\begin{itemize}
\item[-$R_1$] (delete). The relationship between $Z_u^{t-1}$ and $Z_v^{t-1}$ 
for any neighbour
$v$ of $u$ may mean that a delete operation is possible. This may be done according to the two following subrules; if both are applicable, possibly relative to various neighbors, the one deleting the greatest number of letters is chosen. (In case of ties, the choice does not matter.)
\begin{itemize} 
\item[-$R_{1.1}$] If some $Z_v^{t-1}$ is a proper prefix of $Z_u^{t-1}$ and $v$'s 
last action was a $delete$,
$Z_u^t$ is obtained by deleting the last $min\{|Z_u^{t-1}|-|Z_v^{t-1}|,3\}$
letters of $Z_u^{t-1}$.
\item[-$R_{1.2}$] If $Z_u^{t-1}=z0x$ with $x\neq\epsilon$ and some $Z_v^{t-1}=z1y$, 
$Z_u^t$ is obtained by deleting the last
$|x|$ letters of $Z_u^{t-1}$;
\end{itemize}
\item[-$R_2$] (change). If $Z_u^{t-1}=z0$ and some $Z_v^{t-1}=z1y$ then 
$Z_u^{t} = z1$    and
$u$'s state becomes    $follower$ if it is $active$;

\item[-$R_3$] (append). If  for some $v$, $Z_v^{t-1}=Z_u^{t-1}1x$, then
$Z_u^{t}$ is obtained by appending $1$ to $Z_u^{t-1}$; 
\item[-$R_4$] (append). If for some $v$,  $Z_v^{t-1}=Z_u^{t-1}0x$, then
$Z_u^{t}$ is obtained by appending $0$ to $Z_u^{t-1}$;
\item[-$R_5$] (append). If $u$'s state is $active$ and $t \le|Y_u|$,
$Z_u^t$ is obtained by
appending $Y_u[t]$ to $Z_u^{t-1}$;
\end{itemize}

If none of these actions apply, then $Z_u^t=Z_u^{t-1}$ and a $null$ signal is sent. Otherwise, a signal corresponding to the resulting action is sent.
We now prove some properties on Algorithm~$\cal S$ including Corollary \ref{cor:bound3}
which shows that when $R_{1.2}$ is applied, $|x|\le 3$ so that the signal to be sent
is within our set of 7 signals.


\begin{lemma}
\label{lem:delete}
For all $t$, if a node $u$ carries out a $delete$ operation at round $t$,
$u$'s operation at round $t+1$ must be another $delete$ operation
or a change operation.
\end{lemma}

\begin{proof}
By induction on $t$. For $t=1$, it is trivially true since there can be no
delete operation at round $1$.
Suppose $u$ makes a delete at time $t$.
The delete operation  carried out 
at round $t$ on $u$ was made
possible by one or more neighbours of $u$ according to rule
$1$. Let $v$ be one such neighbour.

\begin{itemize}
\item
If $R_{1.1}$ was applied at round $t$ on $u$:
\begin{itemize}
\item  $Z_v^{t-1}$ is a proper prefix of $Z_u^{t-1}$,
\item $v$ did a delete at round $t-1$.
\end{itemize}
Thus 
$Z_u^{t-1}=Z_v^{t-1}d$ (for some
non-empty $d$), and
$Z_u^{t}$ is obtained from $Z_u^{t-1}$  by erasing at most $|d|$
letters at the end, i.e., $Z_u^{t}=Z_v^{t-1} d'$ for some $d'$.

By induction, $v$'s action at round $t$ is another $delete$ operation
or a change.
\begin{itemize}
\item If it is a $delete$ operation, 
then $Z_v^{t}$  is again obtained 
by erasing some letters at the end of $Z_v^{t-1}$
thus it is a proper prefix of   $Z_v^{t-1}$ and 
a proper prefix
of $Z_v^{t-1}d'=Z_u^{t}$ and $R_{1.1}$ 
applies again at $(t+1)$ on $Z_u^{t}$.
\item 
If it is a change operation then $Z_v^{t-1}=w0$ (for some $w$), 
$Z_v^{t}=w1.$ Finally, $Z_u^{t}=Z_v^{t-1}d'=w0d'$
and  either $d'$ is a non empty word and $R_{1.2}$ applies with $y=\epsilon$
on $Z_u^{t}=Z_v^{t-1}d'=w0d'$,
or $d'$ is the empty word and
$R_2$ applies with $y=\epsilon$ on $Z_u^{t}=Z_v^{t-1}=w0$: 
$u$ will do a change at round $t+1$ unless another neighbour makes a delete possible.
\end{itemize}
\item
Otherwise $R_{1.2}$ was applied at round $t$ on $u$: 
\begin{itemize}
\item  $Z_u^{t-1}=w0x=$ with $x\neq \epsilon$,
\item $Z_v^{t-1}=w1y$ for some $v$, and 
\item $Z_u^{t}=w0$ by the delete operation at round $t$.
\end{itemize}
Then:
\begin{itemize}
\item If the operation at round $t$ on $Z_v^{t-1}$ was a delete operation,
according to whether
at least $1y$ is deleted or not, the operation on $Z_u^{t}$ at round 
$(t+1)$ is
a $delete$  or a $change.$
\item If the operation at round $t$ on $Z_v^{t-1}$ is a $change$, an $append$
or $null$, then the operation on $Z_u^{t}$ at round
$t+1$ is also a change (again unless
another neighbour makes a delete possible.)
\qed
\end{itemize}
\end{itemize}
\end{proof}

Lemma~\ref{lem:delete} induces immediately:
\begin{corollary}\label{delete-change}
A sequence of delete operations on a node $u$ ends with a 
change operation on $u$.
\end{corollary}

\begin{remark}\label{rem:active}
While a node $u$ remains active and has
not performed a delete, $Z_u$ cannot be a proper prefix of any $Z_v$.
If $u$ does perform a delete, by Corollary \ref{delete-change}, this delete
will be followed, possibly after other deletes, by a change.
That is the first rules applied to $u$ other than $R_5$ must be a
(possibly empty) sequence of deletes followed by an $R_2$. When this $R_2$ is applied,
$u$ ceases to be active.
\end{remark}

\begin{lemma}
\label{lem:prefix}
For all $t$,
for every vertex $u$, there is a vertex $w$ such that $Z_u^t$ is a prefix of
$\alpha(Id_w)$.
\end{lemma}

\begin{proof}
By induction on $t$.

By Remark \ref{rem:active}, while $u$ remains active,
$Z_u^t$ is a prefix of $\alpha(Id_u)$.
Whenever $Z_u$ changes by rules $R_1$, $R_2$, $R_3$ or $R_4$ as a result 
of a neighbour $v$, $Z_u^t$ is a prefix of $Z_u^{t-1}$ or $Z_v^{t-1}$.
\qed
\end{proof} 

In the following lemma and its proof, $a$ always stands for a single letter,
$0$ or $1$.
\begin{lemma}
\label{lem:voisins}
Let $u$ and $v$ be two neighbours. Let $t$ be a round number.
The words $Z_u^t$ and $Z_v^t$
will always take one of the following forms (up to renaming of $u$ and $v$)
where $p$ and $w$ are words.
\begin{enumerate}
\item $Z_u^t=p$ and $Z_v^t=p$,
\item  $Z_u^t=p$ and $Z_v^t=pw$ with $1\leq |w|\leq 2$,
\item $Z_u^t=p0$  and $Z_v^t=p1a$,
\item  $Z_u^t=p1$ and $Z_v^t=p0w$ and $|w|\leq 3$,
\item $Z_u^t=p$ and $Z_v^t=pw$ and $3\leq |w| \leq 6$ and $u$ performed
a delete in round $t$.
\end{enumerate}
\end{lemma}

\begin{proof}
By induction on $t$.

At round $t=0$, $Z_u^t=\epsilon$ and $Z_v^t=\epsilon$.

Without loss of generality, we will always consider the form given in the
lemma and not the reverse.


We consider the five possible relations between  $Z_u^{t-1}$ and $Z_v^{t-1}$
and show that in each case  $Z_u^t$ and $Z_v^t$  still have one of the five
forms.
\begin{enumerate}
\item $Z_u^{t-1}=p,~Z_v^{t-1}=p$.
Each node may carry out any operation. If each carries
 out the same  operation, we remain in case 1. 
\begin{itemize}
\item If $u$ does a $change$: $p=Z_u^{t-1}=p'0$ and $Z_u^t=p'1$.
\begin{itemize}
\item If  $v$ does a $delete$, the 
$delete$ of $v$ may be $R_{1.1}$ or $R_{1.2}$. 
\begin{itemize}
\item If it is $R_{1.1}$
$Z_v^{t-1} =p = p'0$ and $Z_v^t$ is obtained by truncating 1, 2 or 3 letters from
$p'0$, the same as truncating the same number of letters from $p'1 = Z_u^t$.
Thus $Z_v^t$ is a proper prefix of $Z_u^t$ giving case 2 if 1 or 2 letters were
deleted and case 5 otherwise.
\item If it is $R_{1.2}$ then
$p=Z_v^{t-1}=z0x$ and $Z_{v'}^{t-1}=z1y$ for some neighbour $v'$ of $v$ with
$x\neq\epsilon$. 
It corresponds to case 4, and thus the induction implies that
$y=\epsilon$ and $|x|\leq 3$.
Thus $p=Z_v^{t-1}=p'0=z0x$, $Z_v^t=z0$ and $z0$ is a prefix
of $p'$. Finally, $Z_u^t$
and $Z_v^t$ are linked again by relation 2 or 5.
\end{itemize}
\item If $v$ does  an $append$ or $null$ it is case 4.
\end{itemize}

\item If $u$ does a $delete$.
\begin{itemize}
\item If  $v$ does a $delete$ then one word will be a prefix of 
the other as in case 1 or case 2.
\item If $v$ does an $append$ or $null$ it is case 2 or case 5.
\end{itemize}
\item  In the remaining cases both do $append$ or $null$ it is either case 2 
or case 4.
\end{itemize}

\item $Z_u^{t-1}=p,~Z_v^{t-1}=pw$ and $(1\leq |w|\leq 2)$. 
\begin{itemize}
\item
 If $u$ does a $delete$ any operation on $v$ leaves one
word prefix of the other (again case 1, 2 or 5). 
Now, $|p|-3 \le |Z_u^t| < |p|$ and $|p|-2 \le |Z_v^t| \le |p|+3$ giving the claimed bounds of $2$ and $6$ in cases 2 and 5 respectively.
\item If $u$ does a $change$ and if  $v$ does a $delete$
we are in case 2 or 4. 
If $v$ does any other operation, then we are in case 4.
\item If $u$ does an $append$, it must be an $append1$ if the first bit 
of $w$ is 1 (since $R_3$ has priority over $R_4$).
If $v$ does an $append$  the result will be case 2
if $u$ appended the first bit of $w$ or case 4 otherwise (if the first bit of $w$ is $0$ and $u$ appended $1$).
If $v$ does $null$ the result will be case 2 or 4.
If $v$ does $delete$ the result will be case 1, 2, 4 or 5.
\end{itemize}

\item $Z_u^{t-1}=p0,~Z_v^{t-1}=p1a$. 
The node $u$ will do $change$ or $delete$ giving
$p1$  or a prefix of $p$ respectively, again leaving one word a prefix 
of the other for any operation carried out by $v$.
This gives case 1, 2 or 5.

\item $Z_u^{t-1}=p1,~Z_v^{t-1}=p0w$ and  $ |w|\leq 3$. 
\begin{itemize}
\item If $|w|>0$, $v$ will do a delete of at least $|w|$ bits ($R_{1.2}$). If $u$ does a delete or $v$ deletes more than $|w|$ bits, then one word is a prefix
of the other, leading to cases 1, 2 or 5.
If $u$ appended a bit and $v$ deletes $|w|$ bits, then we get case 3.
If $u$ did null and $v$ deletes $|w|$ bits, then we get case 4.
\item Otherwise $w=\epsilon$, $v$ will do a change 
or a delete leaving one word a prefix of the other. Again
this gives case 1, 2 or 5.

\end{itemize}
\item $Z_u^{t-1}=p,~Z_v^{t-1}=pw$, $3\leq|w| \leq 6$ 
where $u$ has just performed a delete.
Then $v$ will apply $R_{1.1}$ and do a $delete3$, and $u$ a $delete$ operation or
a $change$ operation (Lemma \ref{delete-change}).
leaving case 1, 2 or 5, or case 4 if the last bit
of $p$ is 0 and $u$ does a change.
\qed
\end{enumerate}
\end{proof}

Item 4 of lemma \ref{lem:voisins} is the only one which allows rule $R_{1.2}$ to be applied, leading to:
\begin{corollary}
\label{cor:bound3}
If $R_{1.2}$ is applied then
$0<|x|\leq 3$ and $y=\epsilon$.
\end{corollary}

Lemma \ref{lem:voisins} implies:
\begin{theorem}\label{th:complexity}
Let $G$ be a graph of size $n$ and diameter $D$ such that each node $u$
is endowed with a unique identifier $Id_u$ which is a
non-negative
integer.
Let $X$ be the highest identifier.
After at most $|\alpha(X)|+2 D$ rounds,  algorithm $\mathcal S$ terminates
(that is, after this time no node does any operation other than null)
and for each node $u$, $Z_u=\alpha(X)$.
\end{theorem}

\begin{proof}
Let $u_0$ be the node endowed with the highest identifier. 
Let $k$ be a non negative integer.
By induction on $k$ we prove that after at most $|\alpha(X)|+2 k$ rounds
each node at distance at most $k$ from $u_0$
has $Z=\alpha(X)$.
If $u_0$ ever does a change, let the first round at which this happens be $t$.
Then by Corollary \ref{delete-change} $Z_{u_0}^t$ is lexicographically greater
than $\alpha(X)$ but by Lemma \ref{lem:prefix} $Z_{u_0}^t$ is a prefix of some
$\alpha(Y_w)$ contradicting the fact that $X$ is the highest identifier.
Hence
$u_0$ can never cease to be active,
and so
as long as $|\alpha(Id_{u_0})|>t$,
$u_0$ applies $R_5$ at round $t$.
It follows that the Theorem is true for $k=0$.

For the inductive step,
we assume that each node at distance at most $k$ from $u_0$
has $Z^{|\alpha(X)|+2 k}=\alpha(X)$.
Let $v$ be a node at distance $k+1$ from $u_0$. Let $u$ be a node
at distance $k$ from $u_0$ and neighbour of $v$.
By induction, $Z_u^{|\alpha(X)|+2 k}=\alpha(X)$.
Once a node $u$ has $Z_u=\alpha(X)$, it will never do any operation other than null
because that would lead, possibly after a sequence of deletes, to a $Z_u$
lexicographically greater than $\alpha(X)$ contradicting Lemma \ref{lem:prefix}.
From Lemma \ref{lem:voisins} and knowing
that  $Z_u=\alpha(X)$ where $X$ is the highest identifier,
we deduce that  words $Z_u$ and $Z_v$
will always take one of the following forms 
at round $|\alpha(X)|+2 k$ where $p$ and
$w$ are words and $a$ is the bit $1$ or the bit $0$:
(In all cases except the first $Z_u$ must be the lexicographically greater of the two.)
\begin{enumerate}
\item $Z_u=p$ and $Z_v=p$,
\item  $Z_v=p$ and $Z_u=pw$ with $1\leq |w|\leq 2$,
\item $Z_v=p0$  and $Z_u=p1a$,
\item  $Z_u=p1$ and $Z_v=p0w$ and $|w|\leq 3$,
\item $Z_v=p$ and $Z_u=pw$ and $3\leq |w| \leq 6$ and $v$ has just performed
a delete.
\end{enumerate}
The fifth form is impossible since it would lead to $v$, possibly after a sequence
of deletes, doing a change resulting in $Z_v\succ \alpha(X)$.
The first form has $Z_v$ already equal to $\alpha(X)$.
In the second and third forms, similarly $v$ cannot do a delete (or in the second form
a change) because that would lead
eventually to $Z_v\succ \alpha(X)$, so $Z_v$ becomes equal to $Z_u$ after,
respectively, $|w|$ appends or a change and an append (of $a$).
In the fourth form, $v$ will do a delete of $|w|$ letters followed by a change to arrive at
$Z_u$.

Hence, after at most $|\alpha(X)|+2 (k+1)$ rounds $Z_v=\alpha(X)$ and the
result follows.
\qed
\end{proof}

\section{A Spanning Tree Algorithm}
\label{sec:ST}
This section explains how the computation of a spanning tree may be associated 
to the spreading algorithm $\mathcal S$ by selecting for each node $u$
the edge through which $Z_u$ was modified.

Let $u$ be a node, we add for each neighbour $v$,  a variable 
$status_u^v$ whose possible values are in $\{child,parent,other\}$, which indicates the status of neighbour $v$ at node $u$; initially 
$status_u^v=other$.
The computation of the spanning tree occurs

concurrently with the spreading algorithm $\cal S$ as follows. If $R_2$, $R_3$, or $R_4$ is applied at round $t$ relative to neighbour $v$, then $u$

chooses $v$ as parent (if not already the case).
$v$ is chosen arbitrarily among those of $u$'s neighbours justifying the rule applied.
Then, in addition to the signals of the spreading algorithm (indicating how $Z_u$ was updated), $u$ sends a signal $parent$ to $v$ and a signal $other$ to its previous parent
$v'$
(if different from $v$), and it sets $status_u^{v'}$ to other (so that it never has more than one parent). As a result, $v$ sets $status_v^u$ to $child$
and $v'$ sets $status_{v'}^u$ to $other$.

After  receiving signals from neighbours,
in addition to the computation of the new value of $Z_v$ for each neighbour
$v$ by Algorithm 
$\mathcal S$, $u$ updates $status_u^v$.
Algorithm $\mathcal {ST}$ 
denotes the algorithm obtained with Rules of the spreading
algorithm $\mathcal S$ and actions described just above.

\begin{remark}\label{rem9}
By Remark \ref{rem:active}, the first rules applied to $u$ other than $R_5$ must be a
(possibly empty) sequence of deletes followed by an $R_2$. When this $R_2$ is applied,
$u$ ceases to be active and acquires a parent. Thus a
node has no parent if and only if it  is active.
\end{remark}
\begin{remark}\label{15}
A node has at most one parent.
\end{remark}

The next definition introduces for each node $u$ a
pair
$P_u$ that is
used to prove that the graph induced by all the $parent$ relations has no cycle.
\begin{definition}
Let $u$ be a node, let $t$ be a round number of the spreading
algorithm $\mathcal S$; $M_u^t$ is equal to the maximum of $Z_u^{t'}$ for $t'\le t$,
$f_u^t$ is the minimum $t'$ such that $Z_u^{t'}=M_u^t$ and $P_u^t$ is the pair
$(M_u^t,f_u^t)$.
\end{definition}
Every change in $Z_u$ produces a $Z_u^{t+1}$ which is lexicographically greater than $Z_u^t$, except a delete which produces a $Z_u^{t+1}$ which is a proper prefix of $Z_u^t$.
Hence $Z_u^t = M_u^t$ unless $u$ did a delete at round $t-1$ when $Z_u^t$ is a
proper prefix of $M_u^t$.

\begin{definition}
Let $P_1$ and $P_2$ be two pairs $(M_1,f_1)$ and $(M_2,f_2)$.
$P_1 > P_2$ iff $M_1 \succ M_2$ or $(M_1=M_2$ and $f_1<f_2)$.
\end{definition}
In this order, $P_u^t$ is monotonic non-decreasing in $t$.
The following establishes that any non-active node has a parent with a greater value of $P$. 
\begin{lemma}\label{18}
Let $t$ be a round number. Let 
$u_1$ be a node. Then either $u_1$ is active or there exist 
$(u_i)_{1\leq i \leq p}$  nodes of $G$
such that: for $2\leq i \leq p$ $u_i$ is parent of $u_{i-1}$ and $u_p$ 
is active.
\end{lemma}
\begin{proof}
  If $v$ becomes parent of $u$ at round $t$,
  then $u$ has done a {\it change} or an {\it append} at $t$.
Inspecting rules $R2, R3$ and $R4$ show that  $Z_u^t \prec Z_v<^{t-1}$.
$Z_u^t = M_u^t$ since $u$ did not do a delete at $t-1$ and $Z_v^{t -1} \prec M_v^{t-1}$ 
by the definition of $M$ as a maximum.
Hence $M_u^t \prec M_v^{t-1}$ and round $t$ is the first at which $Z_u^t$ has attained
this value
and so $P_v^{t-1}>P_u^t$. Then $P_v$ remains greater than $P_u$ until $Z_u$ next
increases and a node $v'$ (possibly the same as $v$) becomes parent of $u$.
Hence if $v$ is parent of $u$ at the end of round $t$, $P_v^t > P_u^t$.\qed
\end{proof}

\begin{definition}
We say that algorithm ${\mathcal {ST}}$ terminates when algorithm ${\mathcal {S}}$ terminates,
that is no rules of algorithm ${\mathcal {S}}$ apply at any node.
We denote by ${\mathcal {ST}}(G)$ 
the subgraph of $G=(V,E)$ having $V$ as node set and such that
there is an edge between the node $u$ and the node $v$ 
if $u$ is the parent of $v$ or $v$ is the parent of $u$ when algorithm
${\mathcal {ST}}$ terminates.
\end{definition}

By Theorem \ref{th:complexity},
when Algorithm  $\mathcal {ST}$ terminates there is exactly one $active$ node:
the node with highest identifier.
Now, from Remark \ref{rem9} and \ref{15}, and Lemma~\ref{18}:
\begin{proposition}\label{21}
Let $G$ be a connected graph such that each node has a unique
identifier. Let $u$ be the node with the highest identifier.
When algorithm $\mathcal {ST}$ terminates,
the graph ${\mathcal {ST}}(G)$ is a spanning tree of $G$ with $u$ as root.
\end{proposition}

\section{Termination Detection of Algorithm $\mathcal {ST}$}
\label{sec:STT}
This section presents some actions  which, added to algorithm
${\mathcal {ST}}$, enable  the node with the highest identifier to detect
termination of algorithm ${\mathcal {ST}}$; furthermore, as it is the only one,
when it detects the termination it becomes elected.
Our solution is a bitwise adaptation of the propagation process with feedback introduced in \cite{S} and further formalised and studied in Chapter~6 and~7
of~\cite{Tel}. 

\begin{definition}
Let $v$ be a node. Let $t$ be a round number of the spreading
algorithm. The variable $Z_v^t$ is said to be
well-formed if there exists an identifier $Id$
 such that $Z_v^t=\alpha(Id)$.
\end{definition}

To determine if $Z_v^t$ is well-formed, node $v$ can check whether $\mid Z_v^t\mid = 2j+1$, where $j$ is the number of $1$'s before the first $0$.
Each node $v$ is equipped with a boolean variable 
$Term_v$ which is $true$ iff $v$ and all of its subtrees have terminated.
Whenever a rule of the spreading algorithm is applied to node $v$
or a node $u$ becomes a child of $v$,
the variable $Term_v$ is set to $false$, and a signal is sent to its 
neighbours to indicate that $Term_v=false$. Indeed, this variable can be updated
several times for a same node before stabilizing to $true$. 

We describe an extra rule to be added to the 
$\cal ST$ algorithm in order to allow the node
with highest identifier to learn that it is so by detecting
termination of the spanning tree  algorithm.
This rule is considered {\em after} those of algorithm $\cal ST$ in each round.
Let us denote by $N_v$ the set of neighbours of $v$, and by $Ch_v \subseteq N_v$ those which are $v$'s children. Also recall that we omit the round number in the expression on variables when it is non-ambiguous.

\paragraph{The rule:} Given a node $v$, if ($v$ is follower) and ($Term_v=false$)
and ($Z_v$ is well-formed)
and ($\forall w\in N_v$ $Z_w=Z_v$) and ($\forall w\in Ch_v$ $ Term_w=true$)
then $Term_v:=true$.
Furthermore $v$ sends to his parent a signal indicating that $Term_v=true$.

We denote by \STT
 the algorithm obtained  by putting  together the rules
 of Algorithm $\mathcal {ST}$ and this extra rule for termination detection.

Whenever $Z_v$ changes, a rule of algorithm $\mathcal {S}$ has been applied to $v$
 and so $Term_v$ is set to false. Thus
if $Term_v=true$  then $Z_v$ has the same value  it had when $Term_v$
became $true$ the last time.

We say that algorithm \STT terminates when every node $v$ other than the node $u$
with the greatest identifier has $Z_v=\alpha(Id_u)$ and $Term_v=true$ and so
no node has any actions applicable.

From Theorem~\ref{th:complexity} and Proposition~\ref{21}, we know that ST will terminate, at which time all leaves of the constructed spanning tree will have $Term = true$, and a termination signal takes at most $D$ rounds to propagate to the root. This implies:
\begin{proposition}
  \label{prop:STTterm}
Let $G$ be a graph such that each node has a unique (integer) identifier. 
Algorithm \STT terminates within $D$ rounds after the termination of $\mathcal{ST}$. Furthermore,
if the node $u$ has the highest identifier then,  after a run of 
algorithm \STT,  for each neighbour $v$ of 
$u$ $Z_v=\alpha(Id_u)$  and $Term_v=true$ and 
the node $u$ receives from each node $v$ in $N_u$ the signal 
indicating that $Term_v=true$.
\end{proposition}

The next proposition
establishes
that only the node with highest identifier
can receive a termination signal from all neighbours.

\begin{proposition}\label{prop:term}
Let $G$ be a graph such that each node has a unique identifier.
Let $v$ be a node which has not the highest identifier and 
such that $Z_v=\alpha(Id_v)$ and for each neighbour $w$ of $v$
$Z_w=Z_v$. Then there exists a neighbour $v'$ of $v$ such that 
$Term_{v'}=false$.
\end{proposition}

\begin{proof}
%
%
Suppose that some $v$ which does not have the maximum identifier
does have in some round $Z_v=\alpha(Id_v)$ and for each neighbour $w$ of $v$
$Z_w=Z_v$ and $Term(w)=true$. We will deduce a contradiction.

Write $A$ for $\alpha(Id_v)$.
Define $S(A)$ as the set of nodes $w$ such that $\exists t | Z_w^t=A$
and the $A$-parent of $w\in S(A)\setminus v$ as the node $w'$ which becomes
parent of $w$ when $Z_w$ becomes $A$ (say at round $t$).
Since no $\alpha(Id)$ is a prefix of another $\alpha(Id')$,
we have $Z_{w'}^{t-1}=A$ and
once $Z_w=A$, the next modification of $Z_w$ can only be a delete or a change,
which implies that $Z_w$ cannot become $A$ a second time.
$w'$ is also in $S(A)$ and so following the chain of
$A$-parents from $w$ must end at $v$.
Thus any node in $S(A)\setminus v$ is a $A$-descendant of $v$.

Suppose a node $w\in S(A)$ has $Z_w^{t-1}=A$ and sets $Term_w$ true at round $t$. Then all
neighbours of $w$ have $Z^{t-1}=A$ and so any $A$-child $w'$
of $w$ has not changed
$Z_{w'}$ since it became $A$; so $w'$ is still a child of $w$.
Hence $Term_{w'}^{t-1}=true$.
Repeating this argument, any $A$-descendant of $w$ had $Z^{t'-1}=A$ and
set $Term=true$ at some time $t'<t$.

Since by supposition every neighbour $w$ of $v$ has $Z_w=Z_v=A$ and $Term(w)=true$, 
every node in $S(A)\setminus v$
has $Z^{t-1}=A$ and
sets $Term=true$ at some time $t$. So any neighbour of a node
in $S(A)$ is also in $S(A)$. Since $G$ is connected all nodes of $G$
are in $S(A)$.

In particular the node with highest identifier is in $S(A) \setminus v$,
implying that it has an A-parent and so became a follower during the execution
in contradiction with the fact that it is always active.

\qed
\end{proof}
If the node $u$ with highest identifier, 
 becomes $elected$ as soon as, for each neighbour $v$ of $u$, 
 $Z_v=\alpha(Id_u)$ and $Term_v=true$ 
and it receives 
from each child $v$  the signal 
indicating that $Term_v=true$ we deduce:

\begin{theorem}
\label{th:main}
There exists an election algorithm for graphs $G$
in which each node has a unique integer identifier, using
messages of size $O(1)$ which terminates
after at most $|\alpha(Id_u)|+3 D$ rounds
where $u$ is the node with the highest identifier.
\end{theorem}

The time bound follows from Theorem~\ref{th:complexity} and Proposition~\ref{prop:STTterm}.

\section{Further Related Work}
\label{sec:related}

We provide here a more detailed account of the literature on the leader election problem, which recounts and extends references mentioned in Section~\ref{sec:introduction}. The election problem is fundamental in distributed computing and, indeed, there exists
a vast body of literature on the topic -- see for instance the treatment of this problem in standard books~\cite{Tel,Attiya,Lynch,Santoro} and references therein.
This problem is close to that of spanning tree construction, and it seems that it was first formulated by LeLann~\cite{LeLann}. 
As indicated in \cite{KPPRT15}, some simple questions are still open and therefore this is still a problem which is much alive in the distributed computing community.
Usually, this problem is investigated in one of the following three directions:
\begin{enumerate}
\item Characterisation of (anonymous) graphs for which there  exists a 
deterministic election algorithm;
\item Lower and upper bounds of the time complexity and the message
complexity of deterministic election algorithms
depending on how much is initially known about the graph, 
it is assumed that each node has a unique identifier;
\item Randomised election algorithms for anonymous graphs 
depending on the knowledge on the graph such as the size, the diameter
or the topology (trees, complete graphs...).
\end{enumerate}

For the first item  the starting point is the seminal  work of Angluin
\cite{Angluin} which
 highlights, in particular, the key role of coverings: a graph $G$ is a 
covering of a graph $H$ if there is a surjective homomorphism
$\varphi$ from $G$ to $H$ which is locally bijective (the restriction of
$\varphi$ to incident edges of any node $v$ is a bijection between
incident edges of $v$ and incident edges of $\varphi(v)$).
 More general definitions may be found in \cite{BVfibrations}.
Characterisations of graphs for which there exists an election algorithm
depend on the model.
The first  characterisations were obtained in
\cite{BVelection,YKsolvable,MazurEnum}.
The fundamental tool in 
\cite{BVelection,YKsolvable} is the notion of view: the view from a node
$v$ of a labelled graph $G$ is an infinite labelled tree rooted in $v$
obtained by considering all labelled walks in $G$ starting from $v$.
The characterisation in \cite{MazurEnum} used non-ambiguous graphs:
a graph labelling is said to be locally bijective if vertices with the same
label are not in the same ball and 
have isomorphic labelled neighbourhoods. A graph $G$ is   ambiguous 
if there exists a non-bijective
  labelling of $G$ which is locally bijective.
In \cite{GMMrecog}, authors  prove that  the non-ambiguous graphs, as 
introduced by Mazurkiewicz, are exactly
the covering-minimal graphs.
The main ideas of the election algorithm developed in  
\cite{MazurEnum} were applied 
 to some other models in
\cite{CMasynj,CMsynchro,Csofsem} by adapting the notion of covering.
A characterisation of families of graphs which admit an election algorithm
(i.e., the same algorithm works on each graph of the family)
can be found in \cite{CGM12}.

Concerning the second item, lower bounds or upper bounds for deterministic algorithms
when nodes have a unique identifier which is a non negative
integer of size $O(\log n)$:
\begin{itemize}
\item {\it for the time complexity:}
Peleg presents in \cite{Peleg90} a simple time optimal election algorithm
for general graphs: its time complexity is $O(D)$;
 the size of
messages is $O(\log n)$ thus its bit round complexity is
$O(D\log n)$ and the message complexity is $O(D|E|)$
where $|E|$ is the size of the edge set.
More recently, Kutten et al. \cite{KPPRT15}
prove the lower bound $\Omega(D)$ for the time
complexity in a very general context which contains the deterministic case
studied in this paper. Fusco and Pelc \cite{FP15} show that the 
time complexity
of the election problem is
$\Omega(D+\lambda)$ where $\lambda$ is the level of symmetry
of the graph $G$ (Let $G$ be graph. 
The view at depth $t$ from a node is the tree of all paths
of length $t$ originating at this node. The symmetry of $G$ is the
smallest depth at which some node has a unique view of $G$).
In our case, each node has  a unique identifier thus $\lambda=0$, and we obtain
the same bound as~\cite{KPPRT15}.
\item {\it for the message complexity:}
Gallager \cite{Gallager} presents the first election 
algorithm for general graphs with $O(m+n\log n)$ messages, 
where $m$ is the number of edges.
On the negative side, Burns~\cite{burns1980} prove a $\Omega(n\log n)$ lower bound and Kutten et al.~\cite{KPPRT15} a $\Omega(m)$ lower bound which applies even if $n$ is known and the algorithm is randomized. Put together, both lower bounds yield a matching $\Omega(m + n\log n)$ number of messages. (Santoro~\cite{Santoro84} also proves a $\Omega(m+n\log n)$ lower bound for the more specific problem of finding the maximum ID, in a deterministic setting with $n$ unknown.)
The work presented in \cite{GHS83} had a
great influence on many papers, the time complexity of the algorithm 
is $O(n\log n)$ and the message complexity is optimal  in the worst case.
Optimal message complexity in $O(m+n\log n)$ has been obtained 
also in \cite{Awer87}, in this case the time complexity is $O(n)$, the size 
of message is $O(\log n)$ and the
bit round complexity is $O(n\log n)$.
We can note that very efficient algorithms for both election
and spanning tree computation are presented in \cite{KKM90}.
\end{itemize}
This direction has also been subject to recent developments where the impact of particular knowledge is studied such as~\cite{DP16,GMP17,MP16}. Finally, the tradeoff between time and communication complexity in the case of leader election and spanning tree was considered in~\cite{KK13} in the case of ad hoc networks, when only the neighbors are known to nodes.

Regarding the third item, probabilistic algorithms, a Las Vegas algorithm is one which terminates
with a positive probability (in general $1$) and always produces
 a correct result.
A Monte Carlo algorithm is a probabilistic algorithm which always
terminates; nevertheless the result may be wrong with non-zero
probability.
Some results on graphs having
$n$ vertices are expressed with high probability,
meaning with probability $1-o(n^{-1})$ (w.h.p. for short).
Chapter $9$ of \cite{Tel} and \cite{lavault}
give a survey of what can be done and of
impossibility results in anonymous networks concerning the election problem.
  In particular, no
deterministic algorithm can elect (see Angluin \cite{Angluin}, Attiya
et al. \cite{ASW} and Yamashita and Kameda \cite{YK88}); furthermore,
with no knowledge on the network, there exists no Las Vegas election
algorithm \cite{IR90}. 
In~\cite{KuPPRT15}, Kutten et al. present a leader election algorithm to elect (implicitly) a leader (with high probability) that runs in $O(1)$ time using a sublinear amount of messages, namely $O(\sqrt{n}\log^{3/2} n)$.
Monte Carlo election algorithms for anonymous
graphs without knowledge are presented in \cite{IR90,AM94,SS94}.
They are correct
with probability $1-\epsilon$, where $\epsilon$ is fixed
and known to all vertices. M{\'e}tivier et al.~\cite{MRZ15} presents
 Monte Carlo algorithms
which solve the problems discussed above w.h.p. 
 and which ensure for each node $v$ 
an error probability bounded by $\epsilon_v$ where $\epsilon_v$ is
determined by $v$ in a fully decentralised way. To be more precise,
these 
algorithms ensure an error probability bounded by $\epsilon$
where $\epsilon$ 
is the smallest value among the set of error probabilities
determined  independently by each node.
If the network size is known then Las Vegas election algorithms exist, e.g., in~\cite{IR90}.
Finally, recent works like~\cite{GRS18} explore the role played by other network parameters such as conductance and expansion, and some questions in the same spirit as the ones we addressed in this article regarding tradeoff between time complexity and communication complexity are still open.

\section{Conclusion}
\label{sec:conclusion}
This article focused on the problem of deterministic election in arbitrary networks with unique identifiers. Three complexity measures were discussed in general: time complexity, message complexity, and bit (round) complexity.
It was known that $\Omega(m+n\log n)$ is a lower bound for the number of messages and an algorithm with matching complexity exists.
In~\cite{KPPRT15}, Kutten et al. show that concerning the time
complexity $\Omega(D)$ is a lower bound and~\cite{Peleg90} implies that
$O(D)$ is a tight upper bound. 
For bit (round) complexity, we deduced from~\cite{KPPRT15} and~\cite{DS07}
that  $\Omega(D+\log n)$ is a lower bound and we presented 
an algorithm that matches this bound with a running time of $O(D+\log n)$ bit rounds. This algorithm is the first whose bit round complexity breaks the $O(D \log n)$ barrier, and furthermore, through its optimality in terms of bit rounds, it gives a positive answer to whether optimality can be
achieved
both in time and in the amount of communication, which question was thought to be settled due to the impossibility to satisfy both when messages are (as is frequently assumed) of size $O(\log n)$. As such, our results also make a case for studying the complexity of algorithms through the lenses of bit round complexity. Finally, it could be interesting to explore whether some of the techniques presented in this article are applicable when the size of messages is less constrained, e.g., logarithmic (CONGEST model).

\subsection*{Acknowledgment} We thank the anonymous referees for their many helpful comments on an earlier version of this article.

\bibliographystyle{plain}
\bibliography{election}

\end{document}